\documentclass{article}

\usepackage{authblk}
\usepackage{graphicx,amsmath,amssymb,url,xspace}
\usepackage{amsthm}
\usepackage{comment}
\usepackage{xcolor}
\usepackage{tikz}
\usetikzlibrary{fit}
\usetikzlibrary{arrows,shapes,backgrounds}
\usetikzlibrary{calc,shapes, positioning}
\usetikzlibrary{external}
\usepackage{todonotes}
\usepackage{algorithm}
\usepackage{algorithmic}

\usepackage{textcomp}
\usepackage{hyperref}
\usepackage{cleveref}

\newtheorem{theorem}{Theorem}
\newtheorem{prop}[theorem]{Proposition}
\newtheorem{lem}[theorem]{Lemma}
\newtheorem{defi}[theorem]{Definition}
\newtheorem{cor}[theorem]{Corollary}
\newtheorem{rmk}[theorem]{Remark}

\newcommand\algorithmicprocedure{\textbf{procedure}}
\newcommand{\algorithmicendprocedure}{\algorithmicend\ \algorithmicprocedure}
\makeatletter
\newcommand\PROCEDURE[3][default]{%
  \ALC@it
  \algorithmicprocedure\ \textsc{#2}(#3)%
  \ALC@com{#1}%
  \begin{ALC@prc}%
}
\newcommand\ENDPROCEDURE{%
  \end{ALC@prc}%
  \ifthenelse{\boolean{ALC@noend}}{}{%
    \ALC@it\algorithmicendprocedure
  }%
}
\newenvironment{ALC@prc}{\begin{ALC@g}}{\end{ALC@g}}
\makeatother

\usepackage{xparse,refcount}
\ExplSyntaxOn
  \cs_new_eq:NN \calc \fp_eval:n
\ExplSyntaxOff

\newcommand{\NN}{\mathcal N}
\newcommand{\cNN}{\overline{\mathcal N}}

\newcommand{\unrealized}{inactive }

\graphicspath{{figures/}{../figures/}}

\newcommand{\gm}{\ensuremath{G_{M}}\xspace}
\newcommand{\grb}{\ensuremath{G_{RB}}\xspace}

\newcommand{\grbk}{\ensuremath{{G^k_{RB}}}\xspace}
\newcommand{\rbg}{red-black graph\xspace}
\newcommand{\rbgs}{red-black graphs\xspace}
\newcommand{\rs}{red $\Sigma$-graph\xspace}

\newcommand{\dolloone}{Dollo-$1$\xspace}

\newcommand{\realized}{realized\xspace}
\newcommand{\realize}{realize\xspace}
\newcommand{\realization}{realization\xspace}

\newcommand{\removed}{isolated\xspace}
\newcommand{\remove}{isolate\xspace}
\newcommand{\removedfrom}{isolated in\xspace}
\newcommand{\removing}{isolating\xspace}

\newcommand{\VInt}{\ensuremath{C_{I}}\xspace}
\newcommand{\VUni}{\ensuremath{C_{U}}\xspace}
\newcommand{\VCont}{\ensuremath{C_{C}}\xspace}
\newcommand{\VCR}{\ensuremath{C_{R}}\xspace}
\newcommand{\VCB}{\ensuremath{C_{B}}\xspace}
\newcommand{\VSR}{\ensuremath{S_{R}}\xspace}
\newcommand{\VSB}{\ensuremath{S_{B}}\xspace}
\newcommand{\VSBm}{\ensuremath{S_{B}^m}\xspace}
\newcommand{\VCBm}{\ensuremath{C_{B}^m}\xspace}
\newcommand{\Sm}{\ensuremath{S_7^m}\xspace}

\newcommand{\RecRed}{\textsc{ReductionRecusive}\xspace}
\newcommand{\Rea}{\textsc{Realize}\xspace}

\newcommand{\correction}[1]{#1}
\usepackage{orcidlink}

\begin{document}
\author[1]{Paola Bonizzoni}
\author[1]{Gianluca Della Vedova}
\author[2]{Mauricio Soto Gomez}
\author[2]{Gabriella Trucco}
\affil[1]{Universita degli Studi di Milano-Bicocca, Italia.
\{paola.bonizzoni, gianluca.dellavedova\}@unimib.it}
\affil[2]{Università degli Studi di Milano, Italy.
\{mauricio.soto,gabriella.trucco\}@unimi.it}

%

\title{On recognizing graphs representing Persistent Perfect Phylogenies} 
\date{}
\maketitle

\begin{abstract}
The Persistent Perfect phylogeny, also known as \dolloone,  has been introduced as a generalization of the well-known perfect phylogenetic model for binary characters to deal with the potential loss of characters. 
In  \cite{BonizzoniBDT12} it has been proved that the problem of deciding the existence of a Persistent Perfect phylogeny can be reduced to the one of recognizing a class of bipartite graphs whose nodes are species and characters. 
Thus an interesting question is solving  directly the problem of  recognizing such graphs.
\correction{We present a polynomial-time algorithm for deciding Persistent Perfect phylogeny existence in maximal graphs, where no character's species set is contained within another character's species set.}
Our solution,  that relies only on graph properties,  narrows  the gap between the linear-time simple algorithm for Perfect Phylogeny and the NP-hardness results for the Dollo-$k$ phylogeny with $k>1$.


\end{abstract}

\section{Introduction}
\label{sec:intro}
The perfect phylogeny model is the simplest approach to reconstruct the evolutionary history from characters~\cite{Gusfield91}, and it has many applications in computational biology. 
The instance of the computational problem is a matrix, where each row is associated with a species, each column with a character, and the values in the matrix encode which species have any given character.
The question is to determine whether there exists a tree compatible with the model (in this case, where each character is gained once) and whose leaves correspond to the rows of the matrix.
An equivalent representation is given by bipartite graphs, where the vertex set is partitioned into character nodes and species nodes, and edges correspond to $1$s in the matrix.
This relation allows us to reformulate the problem of deciding whether an input matrix admits a tree representation as the recognition of a graph class~\cite{BBG21,Peer_2004}. 
The characterization of matrices that admit a Perfect Phylogeny, also known as the \textit{four gamete} test, has been exploited to obtain a linear-time algorithm for the problem~\cite{Gusfield91}.
Namely, a matrix has a (directed) perfect phylogeny if and only if it does not have two characters and three species inducing the submatrix with the pairs $(0,1)$, $(1,0)$, and $(1,1)$.
The four gamete test name is due to the fact that, for any two characters, the root of the perfect phylogeny has implied values $(0,0)$.
From a graph recognition perspective, bipartite graphs admitting a directed perfect phylogeny are those that do not contain an induced path of five vertices starting in a species vertex, that is\correction{,} a  $P_5$ also called a $\Sigma$-graph.
Notice that a $\Sigma$-graph is the graph corresponding to the forbidden submatrix of the four gamete test.
In this graph reformulation, a polynomial-time algorithm iteratively finds a \emph{universal} character which is then eliminated from the graph, where a character $c$ is universal if it is adjacent to all species in the connected component of $c$ --- the character elimination is called \emph{character realization} --- until we obtain the empty graph.
Moreover, if a given graph has an induced  $\Sigma$-graph, then no character of the subgraph can be realized, even after realizing some other character.
In other words, a polynomial-time algorithm consists of determining if there exists a sequence of character realizations such that all graphs do not have an induced $\Sigma$-graph and the final graph is empty --- such a sequence of realizations is called a \emph{reduction} of the graph~\cite{BonizzoniBDT12}.  
This notion is fundamental also for the optimal algorithm for the Incomplete Perfect Phylogeny Problem~\cite{BBG21,Peer_2004}, where the input is a binary matrix with missing values that have to be completed so that the resulting matrix has a directed perfect phylogeny.

Although the perfect phylogeny problem has been crucial in computational biology to solve haplotyping problems~\cite{Boniz,Gus02}, most recent applications, mainly in tumor phylogeny inference \cite{bonizzoni2018does,gpps,Kebir,Kuipers13102017},  have increased the interest in the Dollo-$k$ model, which is the generalization of the Perfect Phylogeny to the case where binary characters may be lost at most $k$ times in the tree. 
Binary characters are specified by two states: $0$ and $1$, where $1$ is associated with the gain of the character during the evolutionary history, while $0$ corresponds to the absence of the character.
In the general model, state $0$ can be reached by the loss of the character itself, that is, a change of state from $1$ to $0$. 
From a computational point of view, the Dollo-$k$ model leads to NP-complete decision problems for $k>1$ \cite{Dollo}.
The study of the \dolloone model has been done mainly under the name of Persistent Perfect Phylogeny \cite{BonizzoniBDT12,Gusfield2015,wicke2020}.
In this paper, we analyze the complexity of the \dolloone decision problem by considering its formulation as the recognition of graphs representing instances of the \dolloone decision problem. 
In \cite{BonizzoniBDT12} it has been shown that the \dolloone decision problem can be reformulated as computing a reduction in bipartite graphs, called \emph{red-black} graphs. 
Since its introduction in \cite{BonizzoniBDT12}, a characterization of graphs admitting a Persistent Perfect Phylogeny via a minimal set of forbidden subgraphs appears to be a challenging task. 
Thus, the existence of a polynomial time algorithm for recognizing such graphs based on detecting forbidden substructures is an open problem.
\correction{We progress toward solving this question by proving that recognizing graphs with a Persistent Perfect Phylogeny can be solved in polynomial time when restricted to maximal graphs, that is, graphs where no character's species set is a proper subset of another character's species set.}
The recognition algorithm is based on proving the existence of an ordering of characters under the inclusion relationship w.r.t. to their neighborhood when restricted to a given set of species. 
Based on this ordering, we show the existence of a (restricted) universal character: the realization of such \correction{a} universal character produces a reducible graph, allowing the iterative construction of a reduction. 


\section{Preliminary definitions and results}

\paragraph{\bf Matrix representation.}
The traditional representation of an instance for the phylogeny reconstruction problem is a $n\times m$ binary matrix $M$:
the rows and columns of $M$ are associated with a set $S$ of species and a set $C$ of characters, respectively.  In the matrix $M$,  $M[s,c] = 1$ if and only if the character $c$ is present in the species $s$, otherwise $M[s,c] = 0$. Then the values $0, 1$ represent the two  possible \emph{states} of character $c$.  We also say that species $s$ has the character $c$ if $M[s,c] = 1$ or $c$ is in the set of characters of  $s$. 

A \textbf{phylogenetic tree} for the matrix $M$ is a rooted tree that describes the evolution of the set of species from a common ancestor. 
Starting from the root, which represents the species with all characters at state $0$, characters are \emph{gained} or \emph{lost} along the edges of the tree. More precisely,  a character that changes state from $0$ to $1$ is gained along an edge $(x, y)$ labeled by the character,   and vice versa  is lost when it changes state from $1$ to $0$ along the edge.
The Dollo-$k$ problem asks for a phylogenetic tree where 
(1) each character $c$ is gained   at most once in the entire tree;      
(2) each character $c$ may be lost   at most $k$ times along $k$ edges of the tree; and 
(3) each species $s$ is associated with a  tree node $x$ such that along the edges of the path from the root to  node $x$,  only  characters contained in $s$ are gained and not lost. These characters are exactly those in state $1$ in matrix $M$ for row $s$. 
When $k=0$, then characters can be only gained once but never lost in the phylogenetic tree which corresponds to the Perfect Phylogeny model; while if $k=1$, the tree is a Persistent  Phylogeny. 
Figure \ref{fig:alberoesempio} depicts an example of an input matrix $M$ together with a Persistent Perfect Phylogeny phylogenetic tree for $M$.

An alternative \textbf{graph representation} of the input matrix is the following.
  Given a matrix $M$ on $n$ species and $m$ characters, we define the associated (bipartite) graph $G_M$, where $V(G_M)=S\cup C$, and $E(G_M)=\{ \{s,c\}:\: s\in S,\, c\in C,\, M[s,c]=1\}$.
Given the pair $\{s,c\}\in E(G_M)$, we say that $s$ and $c$ are neighbors, otherwise they  are non-neighbors.  
In other words, the vertex set of the graph $G_M$ consists of the species and the characters, and a species $s$ is connected to a character $c$ only if $s$ has the character $c$.

Given a species $s$, with $C(s)$ we denote the set of characters that are adjacent to $s$ in $G_M$. 
Then we say that a species $s$ \emph{includes} another species $s'$ if \correction{$C(s') \subseteq C(s)$}.
Similarly, for a given character $c$, we denote with   $S(c)$ the set of species that are adjacent to $c$ in the graph $G_M$.
We will say that a character $c$ is \emph{universal} for a set $X$ of species nodes if $c$ is adjacent to all the species of $X$, that is $N(c)\supseteq X$, where $N(c)$ is the set of neighbors of $c$.
On the other hand, two characters $c_1$ and $c_2$ are \emph{independent} if they do not share any species, that is if $N(c_1) \cap N(c_2) =\emptyset$.
%
%
\correction{In this paper, we adopt the graph representation introduced in \cite{bonizzoni_when_2014,bonizzoni2014,Bonizzoni2016}.}






In \cite{BonizzoniBDT12}, the authors characterize the matrices admitting a \dolloone (Persistent Perfect) phylogeny representation using a  graph representation. 
More precisely,  they extend the definition of graph $G_M$ by introducing the notion of \emph{red-black graph}, detailed below,  together with a graph operation on characters, called \emph{realization}. 
This graph and the related graph operations allow the representation of the gains and losses of characters in a phylogenetic tree for the general \dolloone problem.

\paragraph{\bf Red-Black graphs} are associated with instances of a slightly more general version of the problem, where the state in the root of the phylogeny is part of the input and is not necessarily the all-$0$ state. In \rbgs, vertices in $C$ are colored \emph{black} or \emph{red}.
Moreover, edges are black or red and an edge is black if and only if its endpoints are both black.
Formally,  a red-black graph consists of a bipartite graph $\grb =(C \cup S, E_R \cup E_B)$ and a set $R\subseteq C$ of red vertices, where $C$ is the set of characters, $S$ is the set of species, $E_R \subseteq R \times S$ is the set of red edges,  and $E_B \subseteq (C\setminus R) \times S$ is the set of black edges.  

The red-black graph \grb for a node $x$ of the tree $T$ is the graph representation of the  instance solved by the subtree of $T$ with root $x$.
Moreover, the red characters of \grb are exactly the characters that have been gained and not lost on the path from the root of $T$ to $x$.
Indeed, initially $\grb$ is the graph $G_M $ associated with a matrix $M$. Then  the gain  (denoted as label $c^+$ of the edge) or the loss (denoted as label $c^-$ of the edge) of the character $c$   along an edge $(x, y)$ of the tree is encoded by a graph operation on the graph associated with the node $x$, and the result of such operation is the graph associated with the node $y$.

More precisely, in the \rbg associated with a node $x$ of a phylogenetic tree, a character $c$ is adjacent  to a species $s$ via a black edge if the character will be gained in the subtree of $T$ rooted at $x$, i.e. $c$ has state $1$ in the species node $s$. 
Conversely, a red edge between a character node $c$ and a species node $s$ represents the fact  that the character was previously gained in the phylogeny, but it will be lost  in the subtree of $T$ rooted at $x$, since the character will have state $0$ in the species node $s$, i.e. it is persistent in the tree.
Figure~\ref{fig:alberoesempio} represents an example of a phylogenetic tree together with the \rbgs associated with each node in the tree.
%

\paragraph{\bf The realization of a character} represents how a \rbg associated with a node in the phylogenetic tree must be modified when a character is gained or lost.
   Namely, given a \rbg \grb and a character $c\in \grb$ incident only on black edges,  the \emph{realization} of $c$ is the graph obtained from \grb by the following steps:
    \begin{enumerate}
      \item adding red edges joining $c$ to the  species nodes in its connected component that do not have character $c$, 
       \item removing all the black edges incident to $c$, and
        \item  as long as such a character exists, remove all edges incident on a \emph{red universal} character, that is, characters that are adjacent via red edges to all the species in the same connected component.
    \end{enumerate}





 The combination of steps 1 and 2  of the realization of a character $c$ corresponds to the addition of a new node and a new edge labeled by the character $c$ in the phylogenetic tree, i.e. $c$ is \emph{gained}.
 When all the red edges incident to a character $c$ are removed in step (3), that is, $c$ was universal with red edges in its connected component and becomes isolated; in the phylogenetic tree this corresponds to adding a new node with an edge labeled with the loss of the character as $c^-$, in this case we say that the \emph{character is \removed}. 
Finally, if the gain or loss of a character isolates a species node $s$, then in the tree, the newly created node is labeled with $s$ and we say that the \emph{species is \removed}.
Formally, assuming that the \rbg is connected, 

\begin{defi}[Realization of a character]
\label{def:realization}
   Given a connected \rbg $\grb =  (S \cup C, E_R \cup E_B)$ and a character $c\in \grb$ such that   $c$ is incident only to edges in $E_B$,  the \emph{realization of character $c$} is the graph $G' =(S \cup C, E'_R \cup E'_B)$ where: 

  \vspace{-5pt}
 \begin{enumerate}
 \item $E'_B = E_B \setminus \{(c,s_1): (c,s_1) \in E_B\}$, 
 \item $E'_R = E_R \cup \{(c, s): (c,s) \notin E_B\}$, 
\item  $E'_R$ is updated by isolating (that is, removing all incident edges) red-universal characters, until no red-universal character exists.

    \end{enumerate}
  \end{defi}

Figure~\ref{fig:alberoesempio} shows an example of the application of the procedure for constructing a \dolloone phylogeny from a suitable sequence of realizations. 
 To distinguish characters that have been realized  from unrealized ones in a \rbg, we will call \emph{active} those characters that are incident on red edges and \emph{inactive} those characters incident only on  black edges. 




\begin{figure}[h!tb]
    \centering
      \includegraphics[width=.95\linewidth]{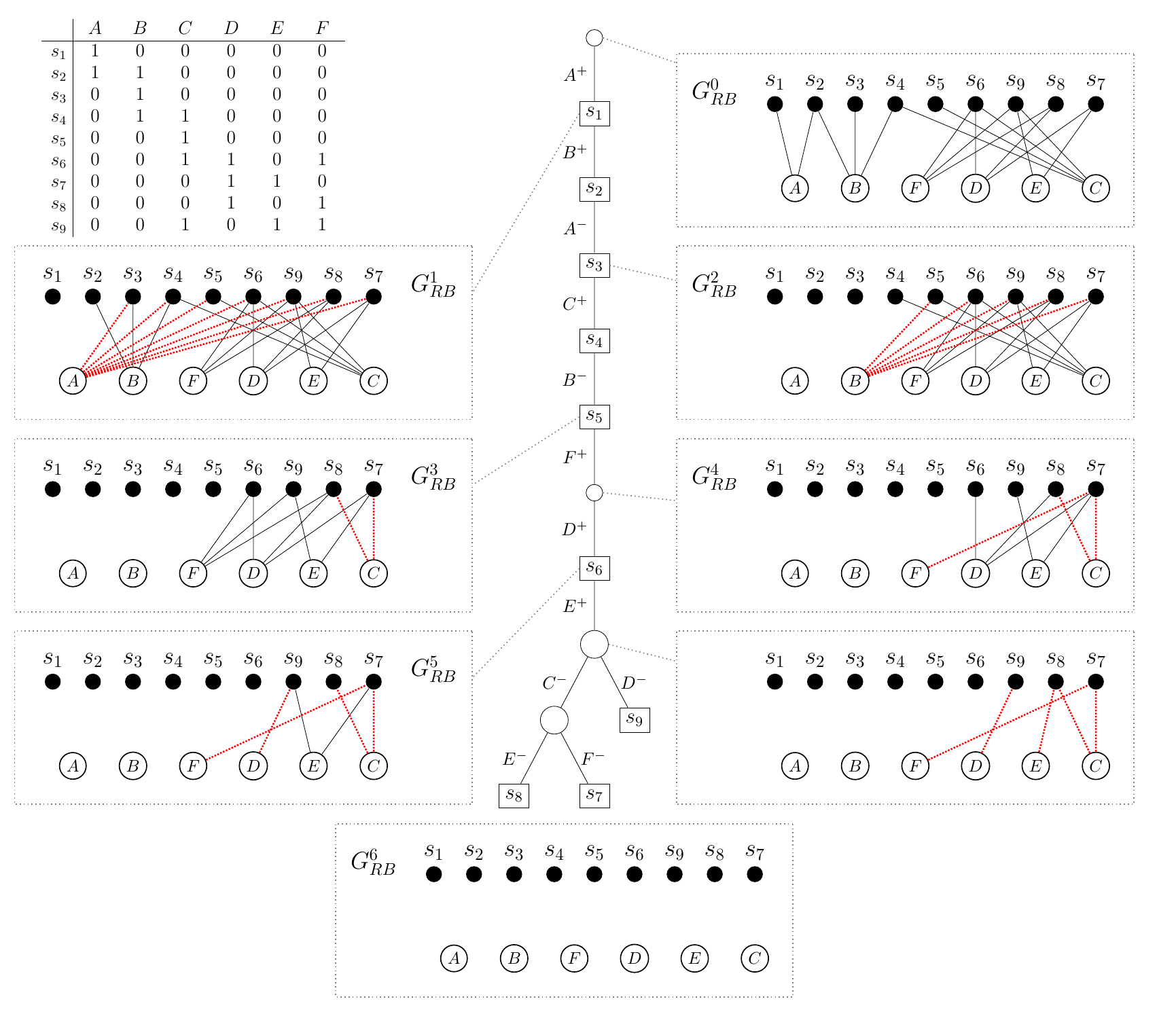}
      \caption{An input binary matrix (top left) and its associated graph (top right). A tree solving the input graph (center) built from the realization of characters according to the reduction $\pi=\langle A,B,C,F,D,E\rangle$. 
    \correction{%
    The \rbgs $G^k_{RB}(\pi)$, $k\in\{1,\ldots,6\}$ are associated with the nodes in the phylogenetic tree and represent partial reductions of the graph.
    Each \rbg $G^k_{RB}(\pi)$ is obtained after the realization of the first $k$ characters according to reduction $\pi$.
    For instance, $\grb^0$ is the black graph encoding the input binary matrix. 
    The graph $\grb^1$ is obtained after the realization of the first character $A=\pi(1)$. 
    In $\grb^1$, all black edges incident to $A$ were removed and red edges were added between $A$ and its non-neighbor species within the same connected component. 
    Note that in $\grb^1$, species $s_1$ becomes isolated and thus labels, in the tree, the node after realizing $A$.
    The realization of $B=\pi(2)$ (the second character of the reduction) isolates the species $s_2$, and makes the node $A$ red-universal, so all red edges incident to $A$ are removed and the node $A$ becomes isolated in $\grb^2$.
    Finally, note that after realizing the last character in the reduction $E=\pi(6)$, the partial reduction becomes non-connected and the tree branches according to the two connected components.
    }
    }  \label{fig:alberoesempio} 
\end{figure}

In \cite{BonizzoniBDT12} the authors prove that a graph has a \dolloone phylogeny if and only if there exists an ordering $\pi=\langle \pi(1),\ldots,\pi(m)\rangle=\langle c_1,\ldots,c_m\rangle$ of its characters  such that the realization of the sequence of characters of $\pi$ reduces the graph to an edgeless one.
More formally, we can state the existence of a \dolloone phylogeny for a graph $G_M$ as follows:

\begin{theorem}[\cite{BonizzoniBDT12}]\label{thm:notredP5} 
A    graph $G_M$   has \dolloone phylogeny if and only if there exists an ordering $\pi=\langle c_1,\ldots,c_m\rangle$ of its characters such that the graph is reduced to an edgeless one. 
\end{theorem}

When realizing characters according to an ordering $\pi$, we generate a sequence of \rbgs that we refer to as \emph{partial reductions}.

\begin{defi}[partial reduction] 
Let $\pi$ be an ordering of the characters of a graph $G_M$.
We denote by $\{\grbk(\pi)\}_{k\in\{0,\ldots,m\}}$ the sequence of \rbgs generated from $G_M$ after realizing the characters according to the ordering $\pi$. 
Starting with $G_M=\grb^0$, the graph $\grbk(\pi)$, called the $k^\text{th}$ \emph{partial reduction of} $G_M$, is obtained after realizing the sequence of characters $\pi(1),\ldots, \pi(k)$.
By an abuse of notation, if the context is clear, we simply denote this sequence by $\{\grbk\}_{k\in\{0,\ldots,m\}}$.
\end{defi}

In the graph $\grbk$, we denote by $\NN^k(c)$ the neighborhood of $c$, and by $\cNN^k(c)$ the set of species that are in the same connected component as $c$, but are not adjacent to $c$.
Note that, before its realization, $c$ is incident only on black edges, and only by red edges thereafter. 
In the later case, we refer to this set as the red neighborhood of $c$. 

In the sequence of partial reductions, $\grb^0=G_M$ is the graph associated with an instance of the problem, and a \dolloone phylogeny exists for the graph if and only if the last partial reduction $\grb^m$ is an edgeless graph. 
In other words, the graph $G_M$ has a solution for the \dolloone problem if all the \rbgs in the sequence of partial reductions can be reduced to a graph with no edges, which motivates the following notion of \emph{reducible graph}.

\begin{defi}[reducible graph]
\label{def:reducible}
A   \rbg \grb is \emph{reducible} if there exists an ordering $\pi$ of its inactive characters, called \emph{reduction}, such that the realization of the characters according to the ordering $\pi$ produces an edgeless graph.  
\end{defi}

\paragraph{\bf Red $\Sigma$-graph.}
As proved in \cite{BonizzoniBDT12}, the characterization in Theorem~\ref{thm:notredP5} can be stated in terms of a forbidden induced subgraph, called a \rs graph, which must not appear in any of the \rbgs in the sequence of partial reductions.
A \rs is a path on red edges composed by two characters and three species. 
The equivalent graph-based characterization of graphs admitting a \dolloone phylogeny is the following:

\begin{lem}[forbidden subgraph]
\label{lem:forbidden}
A   \rbg \grb is  reducible  if there exists an ordering $\pi'$ of its  inactive characters such that the realization of the characters according to the ordering $\pi'$ does not generate a \rbg with an induced \rs.
\end{lem}

This result implies that deciding whether a \rbg \grb is reducible is equivalent to the problem of deciding whether such graphs can be represented as a \dolloone phylogeny. 
In other words, the \dolloone problem can be stated as a graph recognition problem.
Note that the complexity of the problem is not straightforward, since we must guarantee the existence of a sequence of \rbgs avoiding the forbidden structure (a \rs) throughout all \rbgs in the sequence.

Our recognition algorithm builds a reduction one operation at a time by identifying, in a reducible \rbg, an inactive character $c$ whose realization results in another reducible \rbg.
Such a character $c$ is called \emph{safe} in \grb.\\

\begin{defi}[Safe character]
  Given a reducible \rbg~\grb, a character $c$ is \emph{safe} in the graph \grb if its realization results in a reducible \rbg. 
\end{defi}

To simplify the discussion, in the following and without loss of generality, we assume that no two nodes in the graph $G_M$ have the same neighborhood since they represent equivalent species/characters.

\subsection{Structure of the paper and outline of the main results}

The structure of the paper is detailed below by giving an outline of the main properties and related steps that lead to the polynomial algorithm for finding a reduction.

\begin{description}

\item[Section~\ref{sec:extension_properties}] presents some
\textbf{fundamental  properties} that underlie all the main results of the paper.
\correction{Specifically, we describe a necessary condition for a character to be safe and the necessary conditions under which a reduction can be modified to obtain another reduction.} 
A crucial result is Proposition~\ref{prop:universalTopBottom} which states that an inactive character can be realized in a partial reduction  \grbk when it contains or is disjoint from the neighborhood of any other active character.

\item[Section~\ref{sec:maximal}] defines maximal black graphs and shows that in a connected maximal graph \textbf{partial reductions are connected} (\correction{Proposition}~\ref{prop:no_new_components}).

The rest of the section provides the results allowing the construction of a reduction for a \rbg.

\textbf{Section~\ref{sec:Gbr_representation}}  \correction{introduces} the concept of \correction{an} {\bf $S$-partition and \correction{a} $C$-partition} of \correction{the} set of species and \correction{the set of} characters in a \rbg.
Species nodes are partitioned into two sets $\VSB$  and $\VSR$, respectively\correction{,} corresponding to the species that are incident only on black edges, and species that are incident on at least \correction{one} red edge.
Similarly, character nodes are partitioned into $\VCont, \VInt$, and $\VUni$,
which are respectively 
the sets of inactive characters whose set of species is \textbf{contained} in $\VSR$, is \textbf{intersecting} with $\VSR$ and does not contain $\VSR$, and is \textbf{universal} in $\VSR$.
These partitions provide a first rough order of the characters in the construction of a reduction, as proved in Proposition~\ref{prop:order_sets}, which states that characters in \VInt and \VUni must be realized before the ones in \VCont. 

\textbf{Section~\ref{sec:reduction_VInt_VUni}} presents the results to refine the \textbf{order of character in a reduction within the set $\VInt\cup\VUni$}. %
 Proposition~\ref{prop:NN_VB2} proves that a subset of the characters in this set can be ordered according to a containment relation $\pi_U$ which \correction{allows} to individuate the next characters in a reduction.
 This section ends with Theorem \ref{thm:safe_character} summarizing the previous results: 

  \begin{enumerate}
  \item If $(\VInt=\emptyset \wedge \VUni=\emptyset)$ then all characters in $\VCont$ are safe.
  \item If $(\VInt=\emptyset \wedge \VUni\neq\emptyset)$ then a character $c$ is safe if and only if it is safe in the subgraph induced by $\VSB\cup \VUni$. 
  \item If $(\VInt\neq \emptyset)$ then every maximal character of the containment relation $\pi_U$ is safe.
  \end{enumerate}

  \textbf{Section~\ref{sec:initial}} provides some conditions on the set of \textbf{characters starting a reduction}, showing that initial characters can be chosen as the ones belonging to a species of minimum degree (\correction{Proposition}~\ref{prop:first-minimal-species}).

  \item[Section~\ref{sec:algo_1}] \correction{describes} the recognition algorithm and proves its correctness and its polynomial complexity. 

  \item[Section~\ref{sec:conclusions}] concludes the paper discussing the results, applications and future work.

\end{description}





\section{Universal characters and ordering in a reduction}
\label{sec:extension_properties}

The following proposition provides a first necessary condition for the extension of a partial ordering of characters along the construction of a reduction.
%
%
More precisely, it states that to prevent the generation of a \rs, the neighborhood of an inactive character that is realized must contain either the neighborhood or the non-neighborhood of active characters in its same connected component (Proposition~\ref{prop:universalTopBottom}).
%
%
An illustration of this property is given in Figure \ref{fig:universalTopBottom}.

 \begin{figure}[htb]
 \centering
       \includegraphics[width=.5\linewidth]{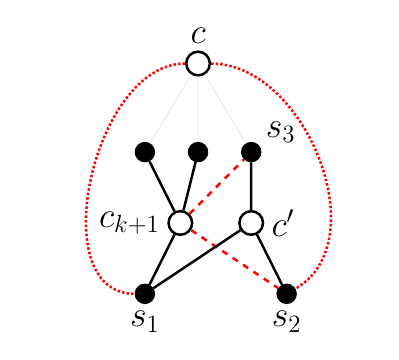}
      \caption{Proposition \ref{prop:universalTopBottom}
        provides a necessary condition for the realization of a character in a given graph \grbk.
       The red dashed edges depict the non-neighborhood of inactive character $c_{k+1}$. 
 Observe that character $c'$ can be realized since its neighborhood contains $\NN^k(c)$. 
 This fact implies that the realization of $c'$ produces additional red-edges in the graph  that are disjoint from those in $\NN^k(c)$, and thus no  \rs is generated. 
 Contrarily, the realization of $c_{k+1}$ generates a \rs, since  $c_{k+1}$  neighborhood does not contain $\NN^k(c)$ nor the complement of $\NN^k(c)$.}
    \label{fig:universalTopBottom}
  \end{figure}

\begin{prop}\label{prop:universalTopBottom} 
  Let $\pi=\langle c_1,\ldots,c_m\rangle$ be a reduction of a graph $G_M$.
  Let $c$ be an active character in \grbk such that $c$ is in the same connected component of $c_{k+1}$, which is the next character to be realized in $\grbk$ according to $\pi$.
  If $c_{k+1}$ has at least one neighbor in $\NN^k(c)$ then $\NN^k(c_{k+1})$ must contain either $\NN^k(c)$ or its complement $\cNN^k(c)$.    
\end{prop}

\begin{proof}
Observe that $\NN^{k}(c)$ is not empty, since by hypothesis, $c_{k+1}$ has at least one neighbor in $\NN^k(c)$, which we denote by $s_1$ \correction{(see Figure~\ref{fig:universalTopBottom})}.
Moreover, $\cNN^{k}(c)$ cannot be empty, since otherwise, $c$ would be a (red) universal \correction{character} in the species set, and it would have been \removed from the \rbg.

Suppose, for the sake of contradiction, that $\NN(c_{k+1})$ contains neither $\NN^{k}(c)$ nor $\cNN^{k}(c)$.
Since $\NN^k(c_{k+1})$ does not contain $\NN^{k}(c)$, there exists a species $s_2$ in $\NN^k(c) \cap \cNN^k(c_{k+1})$ \correction{(see Figure~\ref{fig:universalTopBottom})}.
Now, since $\NN^k(c_{k+1})$ does not contain $\cNN^{k}(c)$ (which is not empty because $c$ is not red universal), it follows that there exists a species $s_3$ in $\cNN^k(c_{k+1}) \cap \cNN^k(c)$.

We conclude that after the realization of $c_{k+1}$, the set $\{s_1, c, s_2, c_{k+1}, s_3\}$ induces a \rs in the graph $G_{RB}^{k+1}$ (see Figure~\ref{fig:universalTopBottom}), which contradicts the fact that the ordering $\pi$ is a reduction of the graph.
Thus, $\NN^k(c_{k+1})$ must contain either $\NN^k(c)$ or $\cNN^k(c)$.
\end{proof}

 \begin{figure}[htb]  
 \centering
       \includegraphics[width=.5\linewidth]{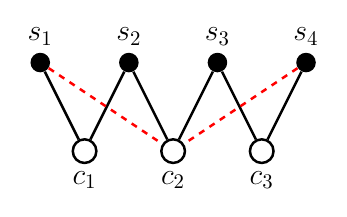}
  \caption{In the figure, white nodes represent characters, black nodes represent species, and red dashed edges depict the non-neighborhood of the \unrealized character $c_2$. 
  In an induced path $P_7$ containing four species and three characters, the central character can not be the first one to be realized.   
}
       \label{fig:centerP7}
  \end{figure}


A direct consequence of Proposition~\ref{prop:universalTopBottom} is the following corollary, illustrated in Figure~\ref{fig:centerP7}.

\begin{cor}\label{cor:centerP7}
  Let $P_7=\{s_1\, c_1\, s_2\, c_2\, s_3\, c_3\, s_4\}$ be a set of four species and three characters inducing a path in a graph \gm.
  Then in any reduction of \gm, the central character $c_2$ of the induced path can not be the first to be realized.
\end{cor}

\begin{proof}
  By contradiction, assume that $c_2$ is the first character among $\{c_1,c_2,c_3\}$ to be realized; then neither $c_1$ nor $c_3$ can be realized according to Proposition~\ref{prop:universalTopBottom}.
\end{proof}

Notice that in a reduction of a graph, it is possible to locally modify the ordering of the characters to obtain an alternative and equivalent reduction. 
For example, in Figure~\ref{fig:alberoesempio}, the characters $F$ and $D$ can be swapped to obtain a different reduction of the graph.
The following \correction{lemma} provides a sufficient condition for the swapping of consecutive nodes in a reduction.
When restated in terms of their associated phylogenetic tree, this property comes from the fact that these characters induce a path containing only positive labels.
In this path, no species are realized, and no internal node has more than one descendant.


 

\begin{lem}[swapping]
\label{lemma:swaporder} 
Let $\pi$ be a reduction of a graph $G_M$.  
  For $1 \le k < m$, let $c_{k}$ be the $k$-th character in the reduction $\pi=\langle c_1,\ldots, c_k, c_{k+1},\ldots, c_m\rangle$. 
  Assume that the connected components of $G_{RB}^{k-1}$ and \grbk  contain the same set of vertices, that is, the realization of $c_{k}$ does not generate new connected components and does not \remove any species in $\grbk$.
  Then the ordering $\pi'=\langle c_1,\ldots, c_{k+1},  c_k,\ldots, c_m\rangle$,  where  $c_k$ is  swapped  with $c_{k+1}$, is also a reduction of $G_M$.
\end{lem}

\begin{proof} 
    Let $\grb$  and $\grb'$ be the red-black graphs associated with  the (partial)  orderings  $\langle c_1,\ldots, c_k,c_{k+1}\rangle$ and $\langle c_1,\ldots, c_{k+1}, c_k\rangle$ respectively.
    Since the vertices of the connected components of $\grb^{k-1}$ and $\grb^k$ are equal, we have that the species nodes in $\grb$ that have not been \removed are exactly the ones in $\grb^{k-1}$ but the ones \removed by the realization of both $c_k$ and $c_{k+1}$.
    Notice that these species are also \removed from $\grb'$.
    Thus, $\grb'$ is either equal to $\grb$ or a proper subgraph.

    By the definition of realization, we notice that from this point, each \rbg in the sequence of red-black graphs generated by $\pi'$ is contained in the ones generated by the sequence $\pi$.
    Thus, we know that no \rs is generated by $\pi'$, and therefore it is a reduction of $G_M$.
\end{proof}

 \section{Maximal graphs: computing a reduction}\label{ordering}
\label{sec:maximal}


\begin{defi}[Maximal graphs]
\label{par:maximal_graph}
A graph $G_M$ associated with an input matrix $M$ is {\em maximal} if it contains only maximal characters, that is, characters whose neighborhood in the graph is not contained in the neighborhood of any other character.  
\end{defi} 

In the following, we will consider only maximal connected graphs. 
In this case, we have the following \correction{proposition}.



\begin{prop}
  \label{prop:no_new_components}
    Let $G_M$ be a maximal connected reducible graph.
    Then for any reduction $\pi$ and for every $1\le k <m$, the partial reduction $\grbk(\pi)$ \correction{contains a single connected component} and has at least \correction{one} active character.
\end{prop}

\begin{proof}
We begin by proving that in every partial reduction, there exists at least one active character.
By contradiction, assume that there exists a reduction of $G_M$ such that for a given $k\in [1,m-1]$, the partial reduction $\grb^k$ has no active characters. 
This means that all the characters in the set $\{c_1,\ldots,c_k\}$, together with their species, have already been \removedfrom $\grb^k$. 
Contrarily, characters in the set $\{c_{k+1},\ldots,c_m\}$ are inactive, meaning that none of their species have been \removedfrom $\grb^k$. 
We conclude that these two sets of characters are independent, meaning they do not share any species.
However, this implies that the former graph $G_M$ was not connected, which contradicts the initial hypothesis.

\correction{In order to prove the proposition, let us assume by contradiction that there exists a partial reduction $\grb^k$ containing more than one connected component}. %
Moreover, let $k$ be the smallest value for which this happens.

We know that each connected component in $\grbk$ must have at least one inactive character. Otherwise, there would exist a connected component formed exclusively by active characters that have not been \removed, and thus it would contain a \rs.
On the other hand, by the previous part, there exists one connected component in $\grbk$ containing an active character $c$. 
Since the graph $\grb^{(k-1)}$ \correction{contains a single connected component} by definition, it means that $c$ was universal in each of the other connected components. 
But this contradicts the maximality of the remaining inactive characters in those components.

\end{proof}

\begin{rmk}\label{obs:initial_path}
\correction{In terms of the phylogenetic tree, Proposition~\ref{prop:no_new_components} implies that every phylogenetic tree of a maximal graph must start with a path containing all character gains and potentially some character losses. 
Indeed, since all partial reductions contain a single connected component, the induced tree does not branch until all characters have been gained. 
However, during character realization, some characters may become isolated and thus be lost in the tree.}
\end{rmk}

\correction{For the sake of simplicity, in the following we will say that the partial reductions $\grbk(\pi)$, $1\le k <m$ are connected.
By Proposition~\ref{prop:no_new_components}, they contain a single connected component; therefore, if isolated species and characters are omitted, the resulting induced graph is connected.}


 \subsection{Partition of the character nodes and general structure of a reduction}
\label{sec:Gbr_representation}

In this section, we introduce a representation of \rbgs that allows to state properties used to characterize the order of characters in a reduction.

\begin{defi}[$S$-partition]
\label{def:species-partition}
  Let \grb be a connected \rbg and let $S$ be the set of species of \grb. Then  the $S$-partition of the set $S$ consists of the following two sets:\\
  - $\VSB(\grb)$  is the set of vertices that are incident exclusively on \textbf{black edges}, and\\
  - $\VSR(\grb)$  is the set of vertices that are incident on \textbf{at least one red edge}.
\end{defi}  
 If the context is clear, and by abuse of notation, we will denote these sets by  $\VSB$ and  $\VSR$ respectively.

\begin{defi}[$C$-partition]
\label{def:char-partition}
Let \grb be a connected \rbg and let $\VCR(\grb)$ be the set of active characters, while $\VCB (\grb)$ is the set of \unrealized characters. Then  the $C$-partition of the set of inactive characters $\VCB (\grb)$ consists  of the sets:  
    \begin{itemize}
        \item[(a)]  $\VCont (\grb)$ is the set of characters whose neighborhood is \textbf{contained} in the set $\VSR$,
        \item[(b)] $\VInt (\grb)$ is the set of  characters with a neighbor and a non-neighbor in \VSR;  we say that such characters \textbf{intersect} both sets \VSB and \VSR,   and   
        \item[(c)] $\VUni (\grb)$  is the set of characters  with a neighbor in $\VSB$ and that are \textbf{universal} for the set  $\VSR$.
        \end{itemize}
\end{defi}

As before, if the context is clear, we denote these sets by $\VCont,\VInt,\VUni$ respectively. 
Similarly, $\VCR$ and $\VCB$ will denote respectively the set of active and inactive characters.

Figure~\ref{fig:general_Grb} depicts the general structure of a \rbg according to the defined partition.
\correction{Moreover, in Table~\ref{tab:alg:execution} can be found the definition of these sets for some of the partial reductions depicted in the example of Figure~\ref{fig:alberoesempio}.}  


 \begin{figure}[htb]  
  \centering
    \includegraphics[width=.65\linewidth]{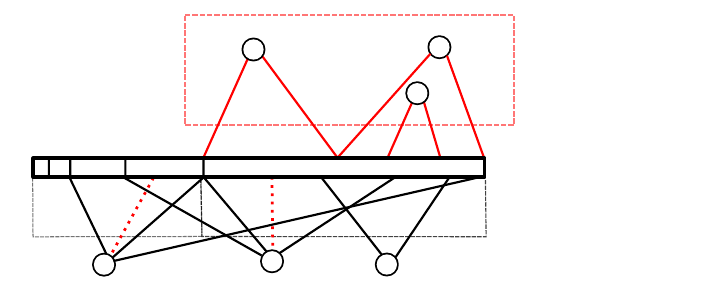}
  \caption{ In the figure, white nodes represent characters, while black and red blocks represent subsets of species.
  The red dashed edges depict the presence of a non-neighborhood of an \unrealized character.
  In a general connected \rbg, we can partition the set of species vertices into the ones incident on a red edge ($\VSR$), and the ones incident only on black edges ($\VSB$).
  Inactive characters can be partitioned according to their neighborhood in the set \VSR into the sets $\VCont, \VUni$ and $\VInt$.
    }
    \label{fig:general_Grb}
  \end{figure}

The next sequence of results aims to establish an ordering, in a reduction, between the sets of nodes in the $C$-partition.
For this purpose, we will focus on the sequence of partial reductions $\{\grbk\}_{k\in\{1,\ldots,m\}}$ of a maximal connected graph.

    
\begin{rmk} \label{rem:character_basics}

The following observations are a direct consequence of the definition of the partition of a \rbg.
    \begin{enumerate}
    \item \label{rem:VRadjVB} In a connected \rbg \grbk, all characters in $\VCR$ were adjacent, before their realization, to  all the  species in $\VSB$.
    \item \label{rem:red_perfect} In a \rbg \grbk, the set $\VCR$ together with its neighborhood $\VSR$ do not contain any \rs and therefore the induced subgraph is solved  by a (red) Perfect Phylogeny. This is a main consequence of the fact that it is a reducible subgraph consisting only of red edges, thus it cannot be solved by a phylogeny with persistent characters. Indeed, we apply Theorem \ref{thm:notredP5} to show that the subgraph has a reduction and admits a perfect phylogeny.
    \end{enumerate}
\end{rmk}


  
  
The following proposition provides a set of relations between sets in a $C$-partition.

\begin{prop}\label{prop:character_relations}
For any partial reduction $\grbk$ of a maximal connected graph, it holds:
  \begin{enumerate}
  \item \label{item:VB_in_VR} each \unrealized character $c\in \VCB$  has a neighbor in the (red) neighborhood of each \correction{of the} active character in \correction{\VCR;}
  \item \label{item:VC_or_VU} either $\VCont$ or $\VUni$ must be empty; 
  \item \label{item:VU_in_VI} each character in $\VUni$ has a non-neighbor in the neighborhood of each character in $\VInt$.
  \end{enumerate}
\end{prop}

\begin{proof}\phantom{a}
  \begin{enumerate}
  \item Let us first prove property \ref{item:VB_in_VR}. By contradiction, assume that an \unrealized character $c$   has no species in the (red) neighborhood of a character $c'\in \VCR$. 
  Therefore before the realization of $c'$ in the graph, $S(c)$ is included in $S(c')$, contradicting the maximality of character $c$.
  \item Let us prove \ref{item:VC_or_VU}. Assume that there exists $c_C \in  \VCont$ and  $c_U \in \VUni$, i.e. both sets are not empty; then $S(c_C) $ is contained in $S(c_U)$ by definition of $\VCont$ and $\VUni$, thus contradicting the maximality of $c_C$.
    \item Let us prove \ref{item:VU_in_VI}. If, by contradiction, we suppose that a character $c_U \in \VUni$ contains all the species of a character $c_I\in \VInt$, then $c_I$ would not be maximal. 
  \end{enumerate}
\end{proof}

The following proposition describes the order in the reduction  of sets in  the $C$- partitions.

\begin{prop}\label{prop:order_sets}
    Let $\pi$ be a reduction of a reducible connected maximal graph $G_M$, and let
    \grbk be a partial reduction of $\pi$.
    Then:  
  \begin{enumerate}
  \item\label{item:B1B2} if $\VInt\cup \VUni \neq \emptyset $, then all
        characters in $\VInt\cup \VUni$ must be realized in the reduction before any active character is \removedfrom the graph, otherwise
  \item\label{item:B0} if $\VInt\cup \VUni=\emptyset$ then all characters in $\VCont$ must be realized before any active character and any species in $\VSR$ is  \removed. 
  Moreover, the characters in $\VCont$ can be realized in any arbitrary order.
  \end{enumerate}  
\end{prop}

\begin{proof}
    Recall that by Proposition~\ref{prop:no_new_components},  all red-black graphs generated along the realization of the characters in a reduction consist of a single connected component.  
    Additionally, we have that an active  character  can be \removed only if it becomes (red) universal.

    If  $\VInt\cup \VUni\neq \emptyset$, then active characters  can become (red) universal only if $\VSB$ is empty.
    Therefore, species in \VSB must be realized before any active character becomes \removed, i.e.  characters in $\VInt\cup \VUni$ must be realized, proving the first statement.
    
  Let us now prove statement \ref{item:B0}.
    If $\VInt\cup \VUni=\emptyset$  then it must be that  $\VSB=\emptyset$.
    In this case, the set of species in the graph is $\VSR$, while all the remaining \unrealized characters are in $\VCont$.
    We have that $\VCR\cup\VSR$ contains more than one (red) connected component; otherwise there exists a red universal character which has not been \removed, leading to a contradiction.
    Moreover, each character in \VCont has a neighbor in each of the (red) neighborhoods of \VCR (Proposition~\ref{prop:character_relations}).
    We conclude that, until the realization of all characters in \VCont, the realization of a character $c_C\in\VCont$ can not create a new connected component and can not \remove any character or species. Thus, their realization can be done in any order according to Lemma~\ref{lemma:swaporder}.

      
  \end{proof}

\begin{rmk}\label{rem_VCont_safe}
    If $\VInt\cup \VUni=\emptyset$ then, by \correction{Proposition}~\ref{prop:order_sets}, all characters in $\VCont$ are safe.
\end{rmk}


\subsection{Constructing a reduction within sets $\VInt$ and $\VUni$}\label{sec:reduction_VInt_VUni}

 While the previous results give the order of realization of inactive characters in $\VCont$ when $\VInt \cup \VUni = \emptyset$, the following propositions establish the ordering, in a reduction,  of characters within the sets $\VInt$ and $\VUni$.

Proposition \ref{prop:NN_VB1} states that a character $c_k$ in the set \VInt must be, just before its realization at $k$, universal in the set  $\VSB(\grb^{k-1})$, that is a character is realized when $\NN^k(c)\cap \VSB(\grb^{k-1}) = \VSB(\grb^{k-1})$. Clearly, once it is realized, some species in \VSB is \removed, and then the next character to be realized is the one that is still universal in the new set $\VSB$.
 As a consequence of this result, we know that there exists an ordering of characters such that the \correction{neighborhoods} of the characters are in inclusion relationship w.r.t. to the set $\VSB$ of species with only incident black edges.

\begin{prop}\label{prop:NN_VB1}
    Let \grbk be the $k$-th partial reduction of a graph $G_M$ and assume that the set \VInt\!(\grbk) contains at least two characters. 
    Then there exists an ordering $\pi_I(\grbk)=\langle c_{I_1},\ldots,c_{I_{|\VInt|}}\rangle$ of the characters in \VInt, such that for all   $1 \le j < |\VInt|$, 
    $\NN^{k}(c_{I_{j+1}})\cap \VSB(\grbk) \subseteq \NN^k(c_{I_{j}})\cap \VSB(\grbk)$.
\end{prop}

\begin{proof}
    The proof is based on the fact that the realization of characters in \VInt can not generate red edges with species in $\VSB$, otherwise a \rs is created. 
    Therefore, at the time of their realization characters in \VInt must be universal in \VSB, creating an order of containment between them.

    Formally, let $c_1$ and $c_2$ \correction{be} two characters in \VInt\!(\grbk), we will show that their neighborhood, restricted to $\VSB$, are in inclusion relation, from which the desired result follows.

     W.l.o.g., \correction{assume} that $c_1$ is realized before  $c_2$ in a reduction.
     If  the neighborhood of $c_2$ in $\VSB$ is not included in the neighborhood of $c_1$  in $\VSB$,  i.e.  they are not in inclusion relation, it means  that there is a species $s_2$ in $\left( \NN^k(c_2)\cap \VSB\right) \smallsetminus \NN^k(c_1) $ (see Figure~\ref{fig:pi_1_2} left).
    Moreover, by Proposition~\ref{prop:order_sets}, we have that in any reduction, characters $c_1$ and $c_2$ must be realized before \removing any of the characters in $\VCR$.
   
    The definition of the set \VInt ensures the existence of a character $c\in \VCR$ such that $c_1$ has a non-neighbor in $\NN^k(c)$.
    Additionally, we know the existence of a neighbor of $c_1$ in $\NN^k(c)$, otherwise $c_1$ can not be maximal (Proposition~\ref{prop:character_relations}.\ref{item:VB_in_VR}).  
    We conclude that the realization of $c_1$  creates a \rs since it is neither universal in $\NN^k(c)$ nor in $\bar\NN^k(c)$ (Proposition~\ref{prop:universalTopBottom}), leading to a contradiction. 
    Thus, in the set \VSB, the neighborhood of $c_1$ must contain the one of $c_2$.
    \end{proof}

\begin{figure}[htb]  
\centering
      \includegraphics[width=.45\linewidth]{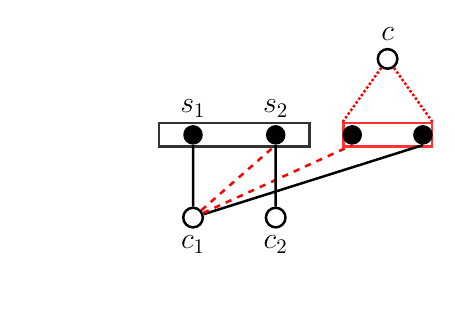}
  \hspace{1em}
        \includegraphics[width=.45\linewidth]{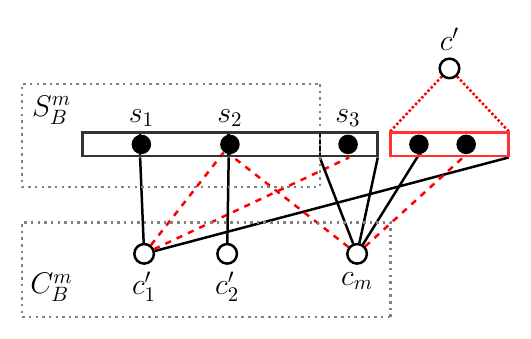}
  \caption{Left: Proof of Proposition~\ref{prop:NN_VB1}. 
  Characters in the set \VInt must be, at their realization, universal in the set $\VSB$. This implies the existence of an inclusion ordering between their neighborhoods.
  Right:  Proof of Proposition~\ref{prop:NN_VB2}. 
  The neighborhood containment relation can be extended to the set $\VInt\cup\VUni$ when restricted to the set $\VSBm$.}
  \label{fig:pi_1_2}
  \end{figure}


  Let $c_m$ denote the last element according to the ordering $\pi_I$ defined in Proposition~\ref{prop:NN_VB1}, and let $\VSBm$ denote the complement, in the set $\VSB$, of the neighborhood of $c_m$, that is, $\VSBm= \VSB\cap \cNN^k(c_m)$ \correction{(see Figure~\ref{fig:pi_1_2} on the right)}.
  Furthermore, let $\VCBm$ denote the union of the characters in \VInt with the universal characters in $\VUni$ with at least \correction{one} neighbor in the set \VSBm (see Figure~\ref{fig:pi_1_2} \correction{right}).
  The following proposition proves that the containment order $\pi_I$ can be extended to include the characters in the set \VCBm.
  
\begin{prop}\label{prop:NN_VB2}
    Let \grbk be the $k$-th partial reduction of a graph $G_M$ such that the set $\VInt$ is not empty.
    Then there exists an ordering  $\pi_U(\grbk)= \langle c'_{B_1},c'_{B_2},\ldots,c'_{B_{|\VCBm|}} \rangle$ of   characters in $\VCBm$  such that
    $
    \NN^k(c'_{B_{j+1}})\cap \VSBm \subseteq \NN^k(c'_{B_{j}})\cap \VSBm, \text{ for all }\: 1\le j< |\VCBm|. 
    $ 
\end{prop}

\begin{proof}
  Similarly to Proposition~\ref{prop:NN_VB1}, we will prove that for every pair of characters $c_1',c_2' \in \VCBm$, their neighborhoods in  $\VSBm$ are in inclusion relation.
  As in Proposition~\ref{prop:NN_VB1} and  w.l.o.g. we can assume that in a reduction, $c'_1$ is realized before $c'_2$.
  
    Moreover, we can assume that both $c'_1$ and $c'_2$ are not the minimum elements of $\pi_I$, otherwise $\NN^k(c'_{i})\cap \VSBm=\emptyset$ and the result trivially holds.  
  Therefore, in any reduction $c_1'$ and $c'_2$ must be realized before $c_m$, otherwise the realization of $c_m$ would generate a \rs together with a character in $\VCR$ according to Proposition~\ref{prop:universalTopBottom}.

    The proof is similar to the one of Proposition~\ref{prop:NN_VB1} and is based on the fact that the realization of $c_1'$ and $c_2'$ can not generate red edges in the set $\VSBm$. Otherwise, they will create a \rs with $c_m$.

    If by contradiction \correction{the neighborhood of} $c'_1$ in $\VSBm$ does not include the neighborhood of $c'_2$, then there exists a species $s_2\in \VSBm$ such that $s_2\notin \NN^k(c'_1)$ (see Figure~\ref{fig:pi_1_2} right) and $s_2$ is in the neighborhood  of $c'_2$.
    
    
    If $c'_1\in \VInt$, we know that at the time of its realization $c'_1$ must be universal in $\VSB$, and the results holds. 
    Thus, we assume that $c'_1\in \VUni$. 
    By Proposition~\ref{prop:character_relations}.\ref{item:VU_in_VI}, we know that $c'_1$ has a non-neighbor, denoted by $s_3$, belonging to $\NN^k(c_m)$.
    Hence, the realization of $c'_1$ creates the red edges $(c'_1\,s_2)$ and $(c'_1\,s_3)$, that is $c'_1 $ is adjacent to a species of $c_m$ and a species outside the  neighborhood of $c_m$.
    Moreover, after its realization, the character $c'_1$ cannot be become red-universal before the realization of $c_m$.
    
    Since $c_m\in \VInt$ and from Proposition~\ref{prop:character_relations}.\ref{item:VB_in_VR}, there exists an active character $c\in \VCR$ such that $c_m$ has both a neighbor and a non-neighbor $s'$ in $\NN^k(c)$. As a consequence, the realization of $c_m$ produces the red edges $(c_m,s_2)$ as $s_2$ is in \VSBm and the edge $(c_m, s')$ for $s'$ in $\NN^k(c)$.
    This fact   makes  the realization of the character $c_m$ impossible, \correction{ since it will become universal on neither} $\NN^k(c)$ nor $\cNN^k(c)$, leading to a contradiction with Proposition~\ref{prop:universalTopBottom}.
\end{proof}

In other words, characters in $\VCBm$ can be ordered according to the inclusion relation of their neighborhood in the set $\VSBm$. We denote by $\pi_U$ the ordering induced by this containment relation between the elements of \correction{$\VCBm$}. 
The following theorem summarizes the previous results of the section and establishes the procedure to compute a safe character in a \rbg when $\VCR\neq\emptyset$.


\begin{theorem}\label{thm:safe_character}
  Let \grbk be a connected \rbg obtained after the realization of the first $k$ characters in a reduction of a maximal \rbg such that the set $\VCR$ of active characters of $\grbk$ is not empty, that is $\VCR\neq \emptyset$. 

  \begin{enumerate}
  \item\label{item:VCont} If $\VInt=\emptyset$ and $\VUni=\emptyset$ then all characters in $\VCont$ are safe.
  \item\label{item:recursive_VUni} If \correction{$\VInt=\emptyset$} and $\VUni\neq\emptyset$ then a character $c$ is safe if and only if $c$ is safe in the subgraph induced by $\VSB\cup \VUni$. 
 \item\label{item:recursive_VInt} If $\VInt\neq \emptyset$, then every maximal character of the ordering $\pi_U$ is safe.
  \end{enumerate}
\end{theorem}
\begin{proof}\phantom{a}
  \begin{enumerate}
  \item  If  $(\VInt=\emptyset \wedge \VUni=\emptyset)$ then $\VSB=\emptyset$. Moreover, as stated in Remark~\ref{rem_VCont_safe} and by Proposition~\ref{prop:order_sets}.\ref{item:B0}, species in $\VSR$ can be  \removed only after the realization of all characters in $\VCont$. \correction{Consequently, all characters in \VCont are safe and} their order of realization is arbitrary.
  \item \correction{If $(\VInt=\emptyset \wedge \VUni\neq\emptyset)$,} by Proposition \ref{prop:character_relations}.\ref{item:VC_or_VU}, $\VCont$ must be empty. 
  Since by hypothesis $\VInt=\emptyset$, then all \unrealized characters are in \VUni, which by their definition are universal in $\VSR$. 
  We conclude that any potential \rs induced by the realization of the remaining \unrealized characters must be in the subgraph induced by $\VSB\cup \VUni$ \correction{(see Figure~\ref{fig:Cu_restart})}.
  
  \item 
\correction{$(\VInt\neq\emptyset)$.} If $\VSBm\neq\emptyset$ then a maximal element of $\pi_U$ is universal on $\VSBm$, therefore no species can be \removed before its realization, hence it is safe.
On the other hand, if $\VSBm=\emptyset$,  a maximal  characters of $\pi_U$ is universal in $\VSB$. Thus, no species can be \removed before their realization, hence they are safe.  
    \end{enumerate}
\end{proof}


Points \ref{item:VCont} and \ref{item:recursive_VInt} of Theorem~\ref{thm:safe_character} provide a sequence of safe characters, thereby an extension of the reduction of \grbk. 
  On the other hand, point \ref{item:recursive_VUni} describes a scenario where active characters impose no constraints on selecting the next safe character, therefore it is essentially equivalent to a \rbg with no active characters.
Although this scenario could potentially lead to an exponential time complexity, we prove in the next section that such a case cannot occur.

\begin{figure}
\centering
      \includegraphics[width=.95\linewidth]{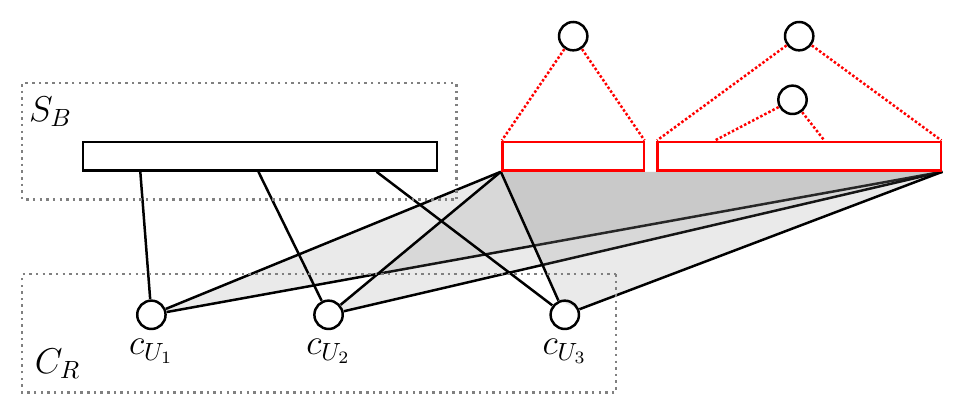}
    \caption{A \rbg where $\VInt=\emptyset$ and $\VUni\neq\emptyset$. In this case $\VCont=\emptyset$, and therefore the \rbg is reducible if and only if the black graph induced by $\VSB\cup\VUni$ is reducible. }
    \label{fig:Cu_restart}
\end{figure}

\subsection{Initial character of a reduction}
\label{sec:initial}

 In this section, we aim to characterize the sequence of characters that start a reduction. 
 The following results guarantee that a reduction of a maximal graph can start with the realization of a species of minimum degree, that is the number of species that have such a character is minimum.

In the following, we will prove a technical Lemma showing that a reduction can be assumed to start by \removing a species. We then prove that this initial species has minimum degree.

Finally, we prove that all minimum degree species other than the initial species one can be \removed only at the end of the reduction. 

 \begin{lem} \label{lemma:first-species}
    For any maximal connected reducible graph $G_M$, there exists a reduction $\pi$ and a species $s_0\in G_M$, such that the first characters in $\pi$ are the characters of species $s_0$.
  \end{lem}

 \begin{proof}
    By Proposition~\ref{prop:no_new_components}, we know that any reduction of $G_M$ does not generate a new connected component until all characters are realized.
    Let $\pi$ be a reduction of a maximal graph $G_M$, and 
    let $k$ denote the first time in the sequence of realizations according to $\pi$ that in the \rbg \grbk a species is \removed.
    
    If $k=m$ then after the realization of all characters, no changes were made in the set of nodes in the former connected component. According to Lemma~\ref{lemma:swaporder}, it is possible to rearrange the order of the realizations to \remove any species in the graph.  
    
    On the other hand, if $k<m$ we have that in the realization of the $k$-th character in the reduction, either a character was \removed or a species was \removed. 
    By definition, a realized character can be \removed only if it becomes universal, and therefore after \removing all its neighbor species, which would contradict the definition of $k$. We conclude that the reduction must begin by \removing a species as required.%

  \end{proof}
  
  \begin{prop}
     \label{prop:first-minimal-species}
   If a maximal graph $G_M$ is reducible, then there exists a reduction starting with all the characters, in arbitrary order, of a minimum degree species node.
  \end{prop}
  \begin{proof}
    According to Lemma~\ref{lemma:first-species}, we can assume that there exists a reduction that starts with the realization of a species $s_0$.
    By contradiction, let \correction{us} suppose that $s_0$ cannot be of minimum degree.
    Therefore, there exists a species $s'$ such that $|\NN(s')|<|\NN(s_0)|$.
    If $\NN(s')\subset \NN(s_0)$\correction{,} then the result holds trivially, since it is possible to realize $s'$ before $s_0$ by rearranging the start of the reduction.
    We conclude that $\NN(s')\nsubseteq \NN(s_0)$, in this case    
    $$
    |\NN(s_0)\setminus \NN(s')|
    =|\NN(s_0)|-|\NN(s_0)\cap \NN(s')|
    > |\NN(s_0)|-|\NN(s')|
    \ge |\NN(s_0)|-(|\NN(s_0)|-1)
    \ge 1.
    $$

    Hence, the set $\NN(s_0)\setminus \NN(s')$ contains at least two characters. 
    Let $c_1$ and $c_2$ denote two of these characters.  
    Moreover, since $c_1$ and $c_2$ are maximal, there exist species $s_1\in \NN(c_1)\setminus \NN(c_2)$ and $s_2\in \NN(c_2)\setminus \NN(c_1)$.
    Finally, since $s_0$ is the first species to be realized in the reduction, we conclude that after the realization of $c_1$ and $c_2$, the set $\{s_1,c_1,s',c_2,s_2\}$ induces a \rs, which is a contradiction.
  \end{proof}

\begin{rmk}
    Notice that after the realization of the initial species $s_0$ in a reduction, the red neighborhood of the character of $s_0$ must be disjoint, otherwise a \rs graph is generated. 
    Therefore, the set $\VSR \cup \VCR$ is composed by $p$ connected components, each of them  formed by the (red) neighborhood of the character of $s_0$. 
\end{rmk}


In the following Proposition~\ref{prop:minimum_species_degree_2}  we show that when a reduction starts with the realization of a species $s_0$ with at least two characters,  we must realize all the characters in the graph before \removing any (red universal) character.

\begin{prop}\label{prop:minimum_species_degree_2}
  Let $G_M$ be a maximal connected \rbg with no active characters ($\VCR=\emptyset$).
  If a reduction of $G_M$ starts with the realization of a species with at least two characters, then all characters of $G_M$ must be realized before any character can be \removed from the graph.
\end{prop}

\begin{proof} 
    Let $s_0$ be the initial species in a reduction, and let $\{c_1,\ldots,c_p\}$ be the set of its characters.
    Since these characters are maximal \correction{but not universal}, we conclude that after their realization, the neighborhood of all characters in $\{c_1,\ldots,c_p\}$ induces a family of non-empty and mutually disjoint sets.
    Therefore, the set $\VCR(\grb^p)$ together with its neighborhood induce a red graph with exactly $p$ distinct connected components \correction{(see Figure~\ref{fig:pcomponents})}.
 
     Notice that, by Proposition~\ref{prop:character_relations}.\ref{item:VB_in_VR}, all \unrealized characters in $\grb^p$ have at least one neighbor in each connected component of $\VCR\cup \VSR$.
     Furthermore, recall that an active character can be \removed only if it becomes universal.
    Therefore, all the inactive characters in $\VCB$ must be realized before any of the characters in $\VCR$ become universal in the set $\VSR$.
    We conclude that no negation is possible before completing the sequence of realizations in the reduction.
    

\end{proof}

\begin{figure}
\centering
      \includegraphics[width=.95\linewidth]{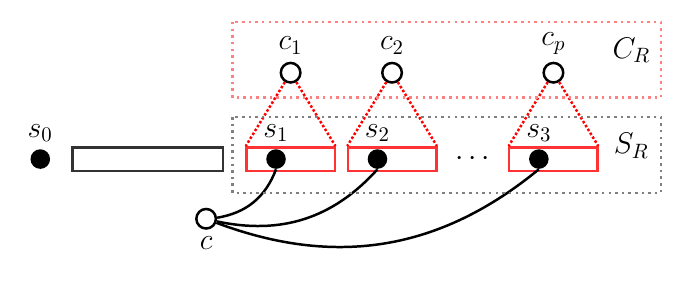}
    \caption{After the realization of the characters $C(s_0)=\{c_1,\ldots,c_p\}$ of an initial species $s_0$, the set $\VCR\cup\VSR$ induces a (red) subgraph containing exactly $p$ connected components. }
    \label{fig:pcomponents}
\end{figure}


Proposition~\ref{prop:first-minimal-species} states that reductions can start with the realization of characters in a minimum size species.
Conversely, we show that all other minimum degree species, beside the initial one, can be \removed only at the end of the reduction. 

The following proposition states that when starting a reduction with a minimum size species, the minimum size species \correction{that do not have as a neighbor a character that is the center of an induced path $P_7$ (as depicted in Figure~\ref{fig:centerP7}) become isolated only at the end of the reduction.
For the sake of simplicity, let denote by \Sm the set of minimum degree species of a graph $G_M$ whose characters do not contain the center of an induced path $P_7$ and thus can be the potential start species for a reduction.}

\begin{prop}\label{minimal_species_in_VSR}
    Let $G_M$ be a \rbg without active characters.
    Then for every reduction starting with $s_0\in \Sm$, none of the species in $\Sm\smallsetminus\{s_0\}$ can be \removed before realizing all the inactive characters in the reduction. 
\end{prop}

\begin{proof}
    Let $p=|C(s_0)|$ be the size of $s_0$.        
    We distinguish two cases: 
    \begin{description}
    \item[Case 1: ($p\ge 2$).] 
      As discussed in the proof of  Proposition~\ref{prop:minimum_species_degree_2}, after the realization of the characters in $C(s_0)$, the \rbg induced by vertices in $\VCR(\grb^p)\cup \VSR(\grb^p)$ contains exactly $p$ connected components \correction{(see Figure~\ref{fig:pcomponents})}.
      
      We claim that none of the species in $\VSB(\grb^p)$ has minimum size, thus all the remaining minimum size  species belong to the set $\VSR(\grb^p)$.
      Indeed, by the Remark~\ref{rem:character_basics}.\ref{rem:VRadjVB}, each character in $C(s_0)$ is adjacent, in the graph $G_M$ (i.e. before their realization), to all the species in $\VSB(\grb^p)$.
      Thus, each species in $\VSB(\grb^p)$ is adjacent to all characters in $s_0$.
      Moreover, the species different from $s_0$ must be adjacent to at least one character distinct from those in $s_0$ as they still have to be realized 
      and are all distinct species.
        It follows that all the species in $\VSB(\grb^p)$ have degree at least $p+1$, and therefore they cannot be of minimum size.  

       Additionally, by Proposition~\ref{prop:minimum_species_degree_2}, none of the realized characters can be \removed before the realization of all the characters in the reduction.
        Therefore, $\Sm\smallsetminus\{s_0\}\subseteq \VSR(\grb^p) \subseteq  \VSR(\grb^{m-1})$.

    \item[Case 2. ($p=1$).] Let $C(s_0)=c_0$, and let $s' \neq s_0$ be another minimum degree species having a single character, which we denote by $c'$.

    Assume that in a reduction, the species $s'$ is \removed before the realization of all the inactive characters; we will prove that in this case $c'$ is in the middle of an induced seven-path, and thus $s' \notin \Sm$.
    
    First, we show that in such \correction{a} reduction, the species $s'$ must be \removed with the realization of $c'$.
    
    Assume to the contrary that instead, when $c'$ is realized, the species $s'$ cannot be \removed.
    Therefore, before \removing $s'$, a character $c''\neq c'$ must be realized.
    By the maximality of character $c''$, we know that there exists a species $s''$ which is in the set of species of $c''$ but not in the one of $c'$.
    On the other hand, $s'$ is not a character of $s''$ since the only character of $s'$ is $c'$.
    We conclude that the realization of $c''$ removes a black edge between $c''$ and $s''$ and creates a red edge between $c''$ and $s'$ (see Figure~\ref{fig:minimalfirst} left).
    But, since all partial reductions are connected (Proposition~\ref{prop:no_new_components}) it must be that neither $c''$ nor $c'$ can be \removed until all the active species are realized, otherwise neither $c'$ nor $c''$ can not become red universal and $s'$ can not be removed, a contradiction.
    Consequently, $s'$ is \removed with the realization of $c'$.
    We denote by $t$ the step in the reduction when $c'$ is realized, that is $\grb^t$ is the first partial reduction where $s'$ has been \removed.

 Since $s'$ has only the inactive character $c'$ in the original graph, it must be that the realization of all characters before the realization of $c'$ generated a red edge between all the active characters and $s'$.
  



Therefore,  we have that in $\grb^{(t-1)}$, the species $ s'$ is \removed by the realization of $c'$ and just after \removing all active characters different from $c'$.
 Figure~\ref{fig:minimalfirst} (right) depicts the structure of $\grb^{(t-1)}$.
Let $c_1$ be the last character \removed from $\grb^{(t-1)}$ before \removing the character $s'$.
Since character $c_1$ is \removed (becomes  red universal) with the realization of $c'$, there must exist a species $s_1 \in \VSB\left(\grb^{(t-1)}\right)$ having the character $c'$ but no $c_1$.
Otherwise, the character $c_1$ would be red universal before \removing the character $s'$ contradicting the assumption that $c_1$ is the last character to be \removed to allow $s'$ to be \removed.

On the other hand, there must exist a species in the original graph, denoted by $s_2$, which contains the character $c'$ but not $c_1$ (see Figure~\ref{fig:minimalfirst} right). 
Indeed, if  such species does not exist, it would follow that after \removing $c_1$ and $s'$, the character $c'$ can be \removed as it becomes red universal in the graph. 
This is not possible, as by \correction{Proposition}~\ref{prop:no_new_components} it would follow  that no inactive character is left in the graph and thus $s'$ will be the  last species to be \removed from the graph, which is a contradiction with our initial assumption. 


Moreover, since species $s_2$ is different from $s'$ but shares the character $c'$ with $s'$, then $s_2$ must be adjacent to at least \correction{one other} inactive character $c_2$ in $\grb^{t}$. 
Furthermore, $c_2$ is a maximal character and hence compared with $c'$, it has a species $s_3$ that is not a species of $c'$ nor $c_1$, as $c_1$ becomes red universal in $\grb^t$ 
(see Figure~\ref{fig:minimalfirst} right).

Finally, notice that by the maximality of $c_1$, we can ensure the existence of a species $s_4\in C(c_1)$  which is not a species of $c'$ nor $c_2$.

We conclude that the set $\{s_4,c_1,s_1,c',s_2,c_2,s_3\}$ induces a black seven-path in $G_M$, centered at $c'$, and then $s'\notin \Sm$, which concludes the proof.

\end{description}
\end{proof}

  \begin{figure}[htb!]
    \centering
          \includegraphics[width=.95\linewidth]{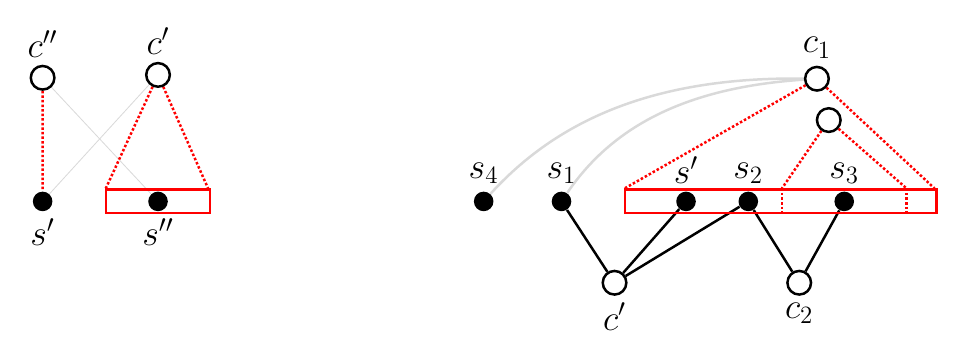}
    \caption{
      Proof of Proposition~\ref{minimal_species_in_VSR}. If a realization starts with a minimum degree species that has a single character, then every other minimum degree species $s'$ with a single character $c'$ must be \removed after all inactive characters are realized, otherwise character $c'$ is the center of a path with seven nodes.}
    \label{fig:minimalfirst}
  \end{figure}

\correction{
\begin{cor}\label{cor:minimal_in_red}
    Let $G_M$ be a \rbg without active characters.
    Then for every reduction starting with $s_0\in \Sm$, we have that all species in $\Sm\smallsetminus\{s_0\}$ are in the set $\VSR(\grb^k)$ for every $k\in\{1,\ldots,n\}$.
\end{cor}
\begin{proof} 
    Let $p=|C(s_0)|$ be the size of $s_0$ and let $s_1\in \Sm\smallsetminus\{s_0\}$ with $C(s_1)=c'$.        
    We distinguish two cases: 
    \begin{description}
    \item[Case 1: ($p\ge 2$).] 
        Since $s_1\neq s_0$ we have that $s_1\in \VSR\left(\grb^p\right)$. 
        Moreover, by Proposition~\ref{prop:minimum_species_degree_2} we have that no character can be isolated before realizing all characters. 
        Therefore, $s_1$ is in $\VSR(\grb^k)$ for every $k\in\{1,\ldots,n\}$.
        
    \item[Case 2: ($p = 1$).] 
        Le us assume by contradiction that there exists $k'\in \{1,\ldots,n\}$ such that in a partial reduction $\grb^{k'}$, the species $s_1$ is not in $\VSR(\grb^{k'})$.
        Therefore, the only active character in $\grb^{k'}$ could be $c'$, and therefore the species $s_1$ has been isolated in $\grb^{k'}$ which is a contradiction with Proposition~\ref{minimal_species_in_VSR}. 
    \end{description}
\end{proof}
}

\begin{rmk}
    Note that the previous result implies that all the minimum size species in $\Sm$, but the initial one, are leaves of the phylogenetic tree generated by the reduction.
    \correction{For instance, the maximal graph $\grb^0$ of the example depicted in Figure~\ref{fig:alberoesempio} contains three minimum size species: $s_1, s_3$, and $s_5$; but only $s_1\in \Sm$.
    Indeed, $s_3$ has as neighbor the character $B$ which is the center of the path $\{s_1\, A\, s_2\, B\, s_4\, C\, s_9\}$, while species $s_5$ has as neighbor the character $C$ which is the center of the path  $\{s_7\, E\, s_9\, C\, s_4\, B\, s_2\}$.
    Thus, the species $s_1$ is the only potential initial species to start a reduction. 
    Moreover, as can be seen the phylogenetic tree in Figure~\ref{fig:alberoesempio} both $s_3$ and $s_5$ do not label any leaves.  
    }
\end{rmk}

\section{Recognizing   maximal reducible graphs in polynomial time}\label{sec:algo_1}

In this section we propose a polynomial time algorithm for recognizing maximal reducible graphs that is based on properties of reductions proved in the previous sections.

Let us recall that, by Proposition~\ref{prop:no_new_components}, partial reductions consist of a single connected component. Thus, the algorithm works under this assumption.  
If this is not the case, we know that a reduction can be constructed independently on each of the connected components of the graph. 

The algorithm iterates the following main steps until the graph has no \unrealized characters or a partial reduction contains a \rs, which means that the graph is not reducible, that is does not represent a \dolloone phylogeny (Theorem \ref{thm:notredP5}). 
Each iteration aims to realize a safe character. 

\correction{
\begin{enumerate}

    \item Initially, the $S$-partition of the species set $S=\VSB\cup\VSR$ and the $C$-partition of inactive characters, that is    $\VCB=\VCont\cup\VInt\cup\VUni$, are computed. 
    The following cases are possible.

    \item \label{item:case_empty_CR} If $\VCR$ is empty, then $\VInt=\emptyset$, $\VCont=\emptyset$ and $\VUni\neq\emptyset$. 
    This scenario occurs, for example, at the initial iteration of the algorithm when dealing with a black graph without active characters.
    In this case, Proposition~\ref{prop:first-minimal-species} guarantees the existence of a reduction starting with a minimum degree species.
    If the set \Sm has a single element, the reduction starts with this species. 
    Otherwise, the set \Sm could have multiple elements. 
    Since the potential initial minimum degree species is unknown, we must iterate over all minimum degree species in \Sm as starting points until we obtain a reduction of the graph.
    This procedure could potentially lead to an exponential time complexity when the iteration through the minimum degree species must be executed multiple times.
    Nevertheless, in Proposition~\ref{prop:no_pendant_VUni} we will show that this iteration must be performed only once in the entire algorithm execution.
    
    \item If $\VInt=\emptyset$ and $\VUni=\emptyset$ then characters in \VCont need to be realized: any of them is safe according to Theorem~\ref{thm:safe_character}~(\ref{item:VCont}).
    
    \item If $\VInt=\emptyset$ and $\VUni\neq\emptyset$, then by Proposition \ref{prop:character_relations}~(\ref{item:VC_or_VU}), we have $\VCont=\emptyset$.   
    In this case, all inactive characters are in \VUni. By  Theorem~\ref{thm:safe_character}~(\ref{item:recursive_VUni}), this case reduces to a scenario where $\VCR=\emptyset$ (see Figure~\ref{fig:Cu_restart}), which has been addressed in the point \ref{item:case_empty_CR}.
    Furthermore, if in any of the previous iterations the procedure has iterated through the species in \Sm, then by Proposition~\ref{prop:no_pendant_VUni}, we have that \VUni is composed of a single character, thus forcing the selection of the next character of the reduction.

    \item Finally, if \VInt is not empty, then   Theorem~\ref{thm:safe_character}~(\ref{item:recursive_VInt}) ensures that a maximal character of the order $\pi_U$ is safe.   Note that the order $\pi_U$ is defined by universal characters over a specific subset of species in \VSB.
                  
    \end{enumerate}
}

The algorithm returns, when it exists, a reduction of the input graph.
Its correctness follows from the fact that it iteratively finds a safe character in all the partial reductions. 
Conversely, if the input graph does not have a reduction, the algorithm returns an ordering that generates a \rs when characters are realized according to this ordering.

Table \ref{tab:alg:execution} depicts the state of the different sets of the character partition along the execution of the algorithm on the graph $G_M$ defined in Figure~\ref{fig:alberoesempio}.

\begin{table}[h!tb]
    \centering
 \setlength{\tabcolsep}{1pt} 
 \begin{tabular}{cc ccccccc}
   \text{\correction{Iteration}} & 
   \begin{tabular}{cc}\text{\correction{Partial}} \\ \text{\correction{reduction}}\end{tabular}&
   \Sm& \VInt & \VUni &   \VCont &
   \text{\correction{$c_m$}} &
   $\pi_U$ &\text{Realization}\\
   \hline
   0 & $\grb^0$ &$\{s_1\}$ & -  & $\{A,B,C,D,E,F\}$ & - & - & - & $A$ \\
   1 & $\grb^1$ &-  & $\{ B \}$ & - & $\{C,D,E,F\}$ & -  & $\langle B\rangle$ & $B$ \\
   2 & $\grb^2$ &-  & $\{ C \}$ & - & $\{D,E,F\}$ & -  & $\langle C\rangle$ & $C$ \\
   3 & $\grb^3$ &-  & $\{F,E\}$ &  $\{D\}$ & -  & $E$ &  $\langle F,D,E\rangle$ & $F,D,E$ \\
 \end{tabular}
    \caption{
    The table depicts the state of the different \correction{relevant} sets along the execution of the algorithm in the instance of Figure~\ref{fig:alberoesempio}.
    \correction{
    We associate each state of the algorithm with its corresponding partial reduction. 
    Note that during the third iteration, we found the situation described in Proposition~\ref{prop:NN_VB2}, where the containment order $\pi_I=\langle F,E\rangle$ of characters in \VInt can be extended to include the characters in \VCBm to define the order $\pi_U$.}
    }
    \label{tab:alg:execution}
\end{table}

\paragraph{\bf Algorithm complexity.}
\label{sec:complexity}

In the following, we will prove that the described  algorithm  has a polynomial time complexity.
To this end, we state the following result which permits to bound the number of iterations  made by the algorithm along its execution.

\begin{prop}\label{prop:no_pendant_VUni}
  Let $G_M$  be a reducible maximal graph with $\VCR=\emptyset$.
  If the set $\Sm$ contains more than one minimum degree species, then  $|\VUni(\grbk)|\le 1$ \correction{for all $1\le k< m$.
  That is, the set of universal characters contains at most one element in all the \rbgs generated by a reduction starting by \removing one of the minimum degree species in \Sm.
}
\end{prop}
\begin{proof}
    \correction{
    Consider a realization starting with a minimum degree species $s_0$. 
    By hypothesis, the set $\Sm$ contains at least two elements; therefore, there exists a minimum degree species $s_1 \in \Sm$ different from $s_0$.
    Clearly, also $s_1$ has $p$ characters, being a minimum degree species as $s_0$, that is $\deg(s_1)=p$.    

    We have two cases:
    \begin{description}
    \item[Case 1: ($p\ge 2$).] 
        Let $\{c_1,\ldots,c_p\}$ be this set of characters of $s_0$. 
        By Corollary~\ref{cor:minimal_in_red} we have that $s_1$ is in the set $\VSR(\grbk)$, for all $1\le k< m$.
        On the other hand, the set $\VCR(\grbk)\cup \VSR(\grbk)$ has at least $p$ connected components, composed  by the non-neighborhood of the $p$ characters of $s_0$ (see Figure~\ref{fig:pcomponents}).
        Now, observe that since $s_1$ is a non-neighbor of $c_1$ and the set $\VCR(\grbk)\cup \VSR(\grbk)$ has at least $p$ connected components, composed by the non-neighborhood of the $p$ characters of $s_0$, it must be that  in the input black graph $G_M$, the species $s_1$ has at least $p-1$ neighbors: $\{c_2,\ldots,c_p\}$ (see Figure~\ref{fig:pcomponents}).
   
        On the other hand, by its definition, for each $k\in \{1,\ldots,m\}$, all characters in $\VUni(\grbk)$ are (black) universal in $\VSR(\grbk)$, hence they are neighbors of $s_1$ in $G_M$. 
        We conclude that 
        \mbox{$\deg(s_1)=p\ge (p-1)+|\VUni(\grbk)|$} 
        and therefore 
        \mbox{$|\VUni(\grbk)|\le 1.$}

    \item[Case 2: ($p=1$).] 
         For all $1\le k< m$, we have that $s_1\in \VSR(\grbk)$ (Corollary~\ref{cor:minimal_in_red}).
         Hence, $s_1$ is a neighbor of each character in  $\VUni(\grbk)$ in $G_M$.
         
        We conclude that 
        \mbox{$\deg(s_1)=1\ge |\VUni(\grbk)|$} and therefore
        \mbox{$|\VUni(\grbk)|\le 1.$}
    
    \end{description}

    }    
\end{proof}


Observe that in a \rbg, the computation of the sets \VCont, \VUni, and \VInt requires time $O(nm) $ in \correction{the} worst case by a naive algorithm based on visiting  the neighborhood of each node in the \rbg. 
The realization of a character requires the computation of the connected components \correction{of in its corresponding partial reduction} and thus in the worst case requires $O(nm)$. 
Since we need to update the sets \VSB and \VSR,  a naive approach would
require  a polynomial time complexity that is $O(n^2 m)$ to update the graph. 

 Moreover, Proposition~\ref{prop:no_pendant_VUni} guarantees that when multiple minimum degree species are found, the input \rbg contains at most one inactive universal character in all the potential subsequent iterative calls.
 Therefore, through the entire execution of the algorithm, we may repeat the realization of a character at most $O(n)$ times, i.e. the number of distinct minimum degree species in a graph.
 
 As mentioned above, \correction{in each iteration}, a naive implementation requires $O(n^2 m)$ time to compute and realize a single safe character and update the graph.
 \correction{Thus, in the worst case, we may have a time complexity that is $O(n^2 m^2)$ for $m$ realizations of all safe characters. 
 Since we repeat the realization of a single character at most $O(n)$ times, the overall complexity for a naive approach is $O(n^3 m^2)$.}

\section{Conclusions}
\label{sec:conclusions}
In this paper, we settle the complexity of recognizing maximal graphs representing \dolloone phylogenies by providing a polynomial algorithm based on an iterative construction of a reduction.
It is worth noticing that \correction{our algorithm exploits the existence of a notion of a universal character restricted to a species subset, similarly to the strategy used in \cite{Peer_2004} for the Perfect Phylogeny problem on incomplete matrices.}
An interesting open question is to provide a characterization of the class of maximal graphs representing a \dolloone phylogenies 
based on a set of forbidden substructures, as in the case of the Perfect Phylogenies.

\correction{Our results encourage the search for a polynomial-time recognition algorithm for the general case where non-maximal characters are included.
However, this generalization presents significant challenges. In the general case, partial reductions are not necessarily connected. 
Moreover, this partition might not be unique, which increases the potential choices in the sequential construction of a reduction and ultimately could increase the problem complexity.
Nevertheless, our method provides a foundation upon which a general solution can be built.}

\section{Acknowledgments}
 
P.B. and  G.D.V. have received funding from
the European Union's Horizon 2020 research and innovation programme under the Marie Sk\l{}odowska-Curie grant
agreement PANGAIA No. 872539.  and ITN ALPACA N.956229.
P.B. is also supported by the grant MIUR
2022YRB97K, PINC, Pangenome Informatics: from Theory to
Applications, funded by the EU, Next-Generation EU, PNRR Mission 4.

M.S.G. is supported by the National Center for Gene Therapy and Drugs Based on RNA Technology—MUR (Project no.CN\_00000041) funded by NextGeneration EU Program.

\bibliographystyle{splncs04}
\bibliography{algolab,persistent}
\end{document}